\documentclass[conference]{IEEEtran}
\IEEEoverridecommandlockouts

\usepackage{amsmath,amssymb,amsfonts}
\usepackage{algorithmic}
\usepackage{algorithm} 
\usepackage{graphicx}
\usepackage{textcomp}
\usepackage{xcolor}
\usepackage{booktabs}
\usepackage{multirow}
\usepackage{subfig}
\usepackage{siunitx}
\usepackage{comment}
\def\BibTeX{{\rm B\kern-.05em{\sc i\kern-.025em b}\kern-.08em
    T\kern-.1667em\lower.7ex\hbox{E}\kern-.125emX}}

\newcommand{\best}[1]{\textbf{#1}}

\newtheorem{proposition}{Proposition}[section]

\newtheorem{assumption}{Assumption}[section]

\begin{document}

\title{\vspace{0.9em}
\linespread{0.8}\selectfont
Decoupling Geometric Planning and Execution\\
in Scalable Multi-Agent Path Finding
\thanks{This work has been supported by the ONR Global grant N62909-24-1-2081, Spanish project PID2024-159284NB-I00 funded by MCIN/AEI/10.13039/501100011033, by ERDF A way of making Europe and by the European Union NextGenerationEU/PRTR, and DGA T45-23R and T64-23R.}
}

\author{
\IEEEauthorblockN{
Fernando Salanova, 
Eduardo Montijano,
Cristian Mahulea
}
\IEEEauthorblockA{
Dept. of Systems Engineering and Computer Science, University of Zaragoza, Zaragoza, Spain\\
Emails: \{fsalanova, cmahulea, emonti\}@unizar.es
\vspace{-1.75em}
}
}

\maketitle

\begin{abstract}
Multi-Agent Path Finding (MAPF) requires collision-free trajectories for multiple agents on a shared graph, often with the objective of minimizing the sum-of-costs (SOC). Many optimal and bounded-suboptimal solvers rely on time-expanded models and centralized conflict resolution, which limits scalability in large or dense instances. We propose a hybrid prioritized framework that separates \emph{geometric planning} from \emph{execution-time conflict resolution}. In the first stage, \emph{Geometric Conflict Preemption (GCP)} plans agents sequentially with A* on the original graph while inflating costs for transitions entering vertices used by higher-priority paths, encouraging spatial detours without explicit time reasoning. In the second stage, a \emph{Decentralized Local Controller (DLC)} executes the geometric paths using per-vertex FIFO authorization queues and inserts wait actions to avoid vertex and edge-swap conflicts. Experiments on standard benchmark maps with up to 1000 agents show that the method scales with an near-linear runtime trend and attains a 100\% success rate on instances satisfying the geometric feasibility assumption. 
Page of the project: https://sites.google.com/unizar.es/multi-agent-path-finding/home

\end{abstract}

\section{Introduction}

Multi-Agent Path Finding (MAPF) is a fundamental problem in robotics, logistics and large-scale system control, including within the Discrete Event Systems (DES). Given multiple agents moving on a shared graph, the goal is to compute collision-free trajectories that avoid (i) \emph{vertex conflicts} (two agents occupying the same vertex at the same time) and (ii) \emph{edge-swap conflicts} (two agents traversing the same edge in opposite directions at the same time), while minimizing the sum-of-costs (SOC), i.e., the sum of goal-arrival times. MAPF is NP-hard due to the combinatorial growth of the joint state space, and optimal centralized solvers such as Conflict-Based Search (CBS) \cite{CBS} rely on explicit coupling across agents, resulting in prohibitive runtime and memory requirements at scale. From a DES perspective, related multi-robot coordination problems have also been addressed using formal supervisory control, e.g., Petri-net-based resource sharing and mutual exclusion~\cite{6648102}.

A key limitation of many MAPF formulations is the dependence on \emph{time-expanded} representations. These approaches discretize time and model each agent assuming synchronous evolution in discrete steps, so collision-avoidance constraints are enforced at every time index. In practice, this implies step-by-step synchronization. Such lockstep execution typically requires a centralized coordination layer (or an equivalent global consensus mechanism) because conflict detection and resolution depend on the simultaneous state of all agents. This modeling assumption is often convenient algorithmically, but it becomes increasingly problematic in large-scale robotic teams, where communication delays, asynchronous actuation, and distributed control architectures make frequent global synchronization costly or impractical.

This paper adopts an execution viewpoint in which agents progress asynchronously along geometric paths and synchronize only when local contention must be serialized. We introduce a two-stage prioritized framework. First, \emph{Geometric Conflict Preemption (GCP)} reduces downstream congestion by penalizing entries into vertices already used by higher-priority paths, favoring inexpensive geometric detours over future waiting. Second, a \emph{Decentralized Local Controller (DLC)} enforces collision-free execution via local per-vertex authorization queues, inserting wait actions only when necessary. The resulting separation between geometric planning and on-demand temporal coordination improves scalability while preserving competitive SOC in constrained environments.

\textbf{Contributions.} The main contributions are:
(i) a congestion-aware prioritized planner (GCP) that inflates \emph{vertex-entry} costs based on higher-priority
geometric footprints, without building a time-expanded graph;
(ii) a queue-based local execution mechanism (DLC) that guarantees collision-free asynchronous execution by
serializing only contested vertex entries; and
(iii) an empirical evaluation up to $1000$ agents showing near-linear scaling trends and improved SOC in
bottleneck-heavy maps via reduced waiting.

\section{Related Work} \label{sec:RW}

MAPF algorithms are commonly grouped into optimal solvers, bounded-suboptimal solvers, and highly scalable (often decoupled) solvers \cite{MAPF-Overview}. Optimal centralized methods, exemplified by \emph{Conflict-Based Search (CBS)} \cite{CBS}, guarantee minimum SOC solutions but can exhibit exponential runtime and high memory usage in dense instances, limiting applicability to large teams. To improve scalability, bounded-suboptimal methods trade optimality for speed while bounding solution quality. \emph{Enhanced CBS}~\cite{ECBS} and \emph{Explicit Estimation CBS}~\cite{EECBS} accelerate CBS using inadmissible heuristics and focal search, but remain centralized and still rely on conflict-tree exploration. As a result, scalability is typically limited to a few hundred agents in constrained environments.

For problems with several hundred agents or more, decoupled methods become essential due to their near-linear \emph{runtime} complexity with respect to the number of agents, since planning is performed sequentially or independently once higher-priority decisions are fixed. \emph{Prioritized MAPF} \cite{P-MAPF} is the representative approach, planning agents sequentially according to a fixed priority ordering. The main drawback is solution quality: higher-priority paths become hard constraints for subsequent agents, often causing cascaded waiting and SOC inflation. This inflation is exacerbated when time is modeled explicitly via space--time reservation tables, which enforce step-by-step synchronization and can introduce repeated waiting actions to resolve conflicts \cite{MAPF_Scalable}. Related challenges appear in continuous-time MAPF, where guaranteeing collision-free motion over real-valued time while maintaining completeness and soundness remains an active topic \cite{Optimal-ContinousTime},\cite{surynek2019multi}.

Our method belongs to the class of scalable, decoupled MAPF approaches, while explicitly targeting congestion mitigation to improve solution quality. In particular, it modifies costs to discourage geometric overlap, in the spirit of congestion-aware planning \cite{EvolutionPaths} and traffic-flow incentives in lifelong MAPF \cite{TrafficFlow}. Unlike iterative replanning or flow-based methods, \emph{Geometric Conflict Preemption (GCP)} applies a one-shot cost inflation derived from higher-priority geometric paths, without constructing a time-expanded graph. Residual conflicts are resolved during execution by a queue-based local controller. The resulting framework remains decoupled and scalable, and it is complete for instances satisfying certain geometric feasibility conditions. 

\section{Problem Definition} \label{sec:PD}
The environment is a connected undirected graph $G=(V,E)$ with nonnegative edge costs
$C:E\rightarrow\mathbb{R}_{\ge 0}$.
Unless stated otherwise, all edges have unit cost.
A set of agents $A=\{a_1,\dots,a_k\}$ is given, where each agent $a_i$ has a start vertex
$s_i\in V$ and a goal vertex $g_i\in V$.
We assume that all starts and all goals are pairwise distinct.

A \emph{path} for agent $a_i$ is a finite sequence of vertices
$p_i=(v_{i,0},v_{i,1},\dots,v_{i,\ell_i})$ such that $v_{i,0}=s_i$ and $v_{i,\ell_i}=g_i$. 
Paths are purely geometric and do not encode temporal information.

To define time execution, conflicts, and SOC, we use the standard discrete-time semantics of MAPF.
A \emph{trajectory} for agent $a_i$ is a function $\tau_i:\mathbb{N}_{\ge 0}\to V$ such that there exists a
finite $T_i\in\mathbb{N}_{\ge 0}$ satisfying $\tau_i(0)=s_i$, $\tau_i(T_i)=g_i$, and $\tau_i(t)=g_i$ for all
$t\ge T_i$. For each $t\ge 0$, either $(\tau_i(t),\tau_i(t+1))\in E$ or $\tau_i(t)=\tau_i(t+1)$.
This timing model is used only to formalize collisions and cost; both planning stages of our method
operate on geometric paths on $G$ nor they construct a time-expanded graph.

A set of trajectories $\mathcal{T}=\{\tau_1,\dots,\tau_k\}$ is \emph{collision-free} if for all
$i\neq j$ and all $t\ge 0$:
\begin{enumerate}
\item $\tau_i(t)\neq \tau_j(t)$ \hfill(\emph{vertex conflict}),
\item $(\tau_i(t),\tau_i(t+1))\neq (\tau_j(t+1),\tau_j(t))$ \hfill(\emph{edge-swap conflict}).
\end{enumerate}
The objective is to find collision-free trajectories minimizing the \emph{sum-of-costs} (SOC):
\begin{equation}
\min_{\mathcal{T}} \; SOC(\mathcal{T}) \triangleq \sum_{i=1}^{k} T_i.
\end{equation}

Most MAPF solvers enforce collision-freeness by planning directly in space--time, which couples agents
at every time index. In contrast, we plan geometrically on $G$ and defer temporal coordination to an
execution mechanism that introduces waiting only when required.

\begin{assumption}\label{ass:residual_reach}
Fix a priority ordering $\pi=(a_1,\ldots,a_k)$ and assume that agents remain at their goals after arrival. For each $i$, let $F_i=\{g_1,\dots,g_{i-1}\}$ and define the residual graph $G_i=(V\setminus F_i,\{(u,v)\in E \mid u,v\notin F_i\})$.
Then, for each agent $a_i$, its goal $g_i$ is reachable from its start $s_i$ in $G_i$.
\end{assumption}

Under Assumption~\ref{ass:residual_reach}, multiple sequential executions exist in which agents move one at a time in any priority order $\pi$ while all others wait, and then remain parked at their goals. Let $p_i^\star$ be a shortest path from $s_i$ to $g_i$ in $G_i$, and let $\ell_i^\star$ denote its cost. This induces a collision-free trajectory set $\mathcal{T}_{\mathrm{seq}}$ with
\begin{equation}
SOC(\mathcal{T}_{\mathrm{seq}}) \;=\; \sum_{i=1}^{k}\ell_{i}^\star \;\triangleq\;
SOC_{\mathrm{seq}}(\pi),
\end{equation}
and therefore $SOC^\star \le SOC_{\mathrm{seq}}(\pi)$, where $SOC^\star$ is the optimal MAPF cost.

\section{Proposed Method} \label{sec:method}

We propose a decoupled, prioritized MAPF framework that separates geometric planning from temporal
coordination, see Algorithm~\ref{alg:gcp_dlc_high_level}. Under a fixed priority ordering $\pi=(a_1,\dots,a_k)$, the method proceeds in two stages:
(i) \emph{Geometric Conflict Preemption} computes a set of time-independent paths
$\mathcal{P}=\{p_1,\dots,p_k\}$ by running prioritized A* on a cost-inflated residual graph; and
(ii) a \emph{Decentralized Local Controller} transforms $\mathcal{P}$ into a collision-free trajectory
set $\mathcal{T}=\{\tau_1,\dots,\tau_k\}$ by inserting waiting actions only when required to resolve local
contention (vertex and edge-swap conflicts) during execution.

\begin{algorithm}[t]
\caption{High-Level Algorithm}
\label{alg:gcp_dlc_high_level}
\begin{algorithmic}[1]
\REQUIRE Connected graph $G=(V,E)$, agents $A=\{a_1,\dots,a_k\}$, starts $\{s_i\}$, goals $\{g_i\}$, edge costs $C$, penalty $C_{\mathrm{p}}>0$
\ENSURE Paths $\mathcal{P}$ and collision-free trajectories $\mathcal{T}$, or \textsc{Fail}

\STATE $\mathcal{P}\gets\emptyset$;\quad $\mathcal{T}\gets\emptyset$
\STATE Initialize queues: $Q_v\gets [\,]$ for all $v\in V$

\COMMENT{GCP: prioritized geometric planning with goal removal and cost inflation}
\FOR{$i=1$ \TO $k$}
    \STATE $F_i := \{g_{1},\dots,g_{i-1}\}$
    \STATE $V_i \gets V\setminus F_i$
    \STATE $E_i \gets \{(u,v)\in E\mid u,v\in V_i\}$
    \STATE $G_i\gets (V_i,E_i)$
    \STATE Define inflated costs $C_i$ on $E_i$ as in \eqref{eq:inflated_cost}
    \IF{$s_i\notin V_i$ \OR $g_i\notin V_i$}
        \RETURN \textsc{Fail}
    \ENDIF
    \STATE $p_i \gets \textbf{A*}(G_i, s_i, g_i, C_i)$
    \IF{$p_i=\emptyset$}
        \RETURN \textsc{Fail}
    \ENDIF
    \STATE $\mathcal{P}\gets \mathcal{P}\cup\{p_i\}$
\ENDFOR

\COMMENT{DLC: offline queue construction from $\mathcal{P}$ and online execution}
\FOR{$i=1$ \TO $k$}
    \FOR{$r=1$ \TO $\ell_i$}
        \STATE $v \gets v_{i,r}$ \COMMENT{the $r$-th vertex visited after the start}
        \STATE Append $a_i$ to the tail of $Q_v$
    \ENDFOR
\ENDFOR

\STATE $\mathcal{T}\gets \textsc{DLC-Execute}(\mathcal{P},\{Q_v\}_{v\in V})$
\RETURN $(\mathcal{P},\mathcal{T})$
\end{algorithmic}
\end{algorithm}



\subsection{Geometric Conflict Preemption (GCP)}
We assume throughout this subsection that Assumption~\ref{ass:residual_reach} holds and a fixed priority ordering $\pi=(a_1,\ldots, a_k)$ is given.
The GCP stage computes a set of time-independent geometric paths sequentially in this order.
Its purpose is not to enforce collision-freeness directly, but to reduce downstream contention by
discouraging overlap with already planned higher-priority paths, without introducing any explicit
time-expanded reasoning.

At step $i$, suppose that geometric paths
$\mathcal{P}_{<i}=\{p_{1},\dots,p_{i-1}\}$ have already been computed.
For a path $p_{j}=(v_{j,0},\dots,v_{j,\ell_{j}})$, define its geometric footprint as
\[
\mathrm{nodes}(p_{j}) := \{v_{j,0},\dots,v_{j,\ell_{j}}\}.
\]

Since agents remain parked at their goals after arrival, goals of higher-priority agents are treated
as permanently occupied and are therefore excluded from subsequent planning.
Accordingly, for agent $a_i$ we plan on the residual graph $G_i$ defined in
Assumption~\ref{ass:residual_reach}.

To mitigate future congestion, we penalize transitions that enter vertices already traversed by higher-priority paths. The penalty is designed to reflect the \emph{expected temporal persistence} of higher-priority agents at those vertices during execution.

For each agent $a_j$ with a geometric path $p_j=(v_{j,0},v_{j,1},\dots,v_{j,\ell_j})$, and for each vertex
$v\in\mathrm{nodes}(p_j)$, define the (first) visit index
\[
r_j(v) := \min\{r\in\{0,\dots,\ell_j\}\mid v_{j,r}=v\},
\]
and set the vertex penalty $C_{\mathrm{p},j,v}:=r_j(v)$.
(If $v\notin \mathrm{nodes}(p_j)$, define $C_{\mathrm{p},j,v}:=0$.)

The choice of the visit index $r_j(v)$ is explained because conflicts are resolved via FIFO queues, a lower-priority agent $a_i$ reaching a vertex $v$ must wait until all higher-priority agents $a_j$ ($j < i$) have visited and emptied $v$. A vertex $v$ located late in a high-priority path $r_j(v)$, effectively remains locked for a longer duration from the perspective of the global start time $t=0$. By penalizing vertices proportionally to their visit index, GCP encourages lower-priority agents to find detours around nodes that are likely to cause long-duration waiting.

For a subsequent agent $a_i$, we define the inflated cost function $C_i: E_i \to \mathbb{R}_{\ge 0}$ by
increasing the cost of entering a vertex according to the penalties induced by higher-priority paths:
\begin{equation}\label{eq:inflated_cost}
C_i(u,v) \;=\; C(u,v) \;+\; \sum_{j=1}^{i-1} C_{\mathrm{p},j,v},
\qquad (u,v)\in E_i.
\end{equation}

We prevent penalizing transitions into goal vertices already reserved for
higher-priority agents, which are removed from $G_i$ and hence cannot be traversed.

Agent $a_{i}$ then computes a shortest path in the residual graph under the inflated costs:
\begin{equation}\label{eq:astar_gcp}
p_{i} \in \arg\min_{p:s_{i}\leadsto g_{i} \text{ in } G_i}
\sum_{(u,v)\in p} C_i(u,v),
\end{equation}
which is obtained by running A* on $G_i$ using any heuristic admissible with respect to $C_i$.
The resulting path $p_i$ is purely geometric and may include small detours that reduce overlap
with already planned paths.

The GCP stage is well-defined under Assumption~\ref{ass:residual_reach}: since $g_i$ is reachable
from $s_i$ in $G_i$ and all edge costs are finite and nonnegative, A* returns a finite-cost path.
Any residual conflicts are spatiotemporal and are handled in the DLC stage by inserting a finite number
of waiting actions, ensuring that every agent eventually reaches and remains at its goal.

\subsection{Decentralized Local Controller (DLC)}

The Decentralized Local Controller (DLC) converts the geometric paths $\mathcal{P}=\{p_1,\dots,p_k\}$ produced by GCP into collision-free trajectories by inserting wait actions only when required by local contention, see Algorithm~\ref{alg:dlc_execute}. DLC enforces mutual exclusion through per-vertex FIFO authorization queues and does not require a global space--time reservation table. While our simulations implement DLC in a centralized loop, the rule depends only on local queue membership and can be implemented in a distributed manner.

For each vertex \(v\in V\), DLC maintains a FIFO authorization queue \(Q_v\).
These queues are constructed offline from the set of geometric paths \(\mathcal{P}\) and the fixed priority ordering \(\pi=(a_1,\dots,a_k)\), where agents are indexed by priority. Specifically, for each agent \(a_i\) and for each vertex \(v\in\mathrm{nodes}(p_i)\), the agent \(a_i\) is appended to the tail of \(Q_v\). As a result, for every vertex that appears in multiple paths, \(Q_v\) encodes a total order of authorized entries that is consistent with the global priority ordering.

During execution, each agent $a_i$ follows its geometric path $p_i=(v_{i,0},v_{i,1},\dots,v_{i,\ell_i})$ and maintains an index $r_i$ indicating its current waypoint $v_{i,r_i}$. If $r_i<\ell_i$, the agent attempts to move from $u=v_{i,r_i}$ to the next vertex $v=v_{i,r_i+1}$. The move $u\to v$ is authorized if and only if $a_i$ is the first element of the FIFO queue $Q_v$. If authorized, the agent executes the move, increments $r_i$, and the authorization is consumed by popping the head of $Q_v$. Otherwise, the agent waits at $u$. Once an agent reaches its goal, it remains there thereafter.

This execution rule enforces mutual exclusion at vertices, since at most one agent can be authorized to enter a given vertex at any time step. Edge-swap conflicts are also prevented implicitly. A simultaneous swap on an edge $\{u,v\}$ would require two agents to be authorized to enter $u$ and $v$ at the same time. However, any two agents whose paths include both vertices appear in both queues
$Q_u$ and $Q_v$ in an order consistent with the global priority. While a higher-priority agent has not completed its visit to a vertex, it remains at the head of the corresponding queue, preventing a lower-priority agent from being authorized to enter the opposite endpoint. Then, opposite traversals on the same edge cannot occur.

Under Assumption~\ref{ass:residual_reach}, each agent has a finite geometric path in the residual graph and all higher-priority agents eventually remain parked at their goals. Consequently, all queues \(Q_v\) are finite and progress monotonically as agents complete their visits. Any temporary blocking is resolved by finite waiting, and the DLC stage always terminates, producing a collision-free trajectory for every agent.

\begin{proposition}[Correctness and Completeness of DLC Execution]
\label{prop:dlc_correctness}
Under Assumption~\ref{ass:residual_reach}, the execution produced by the
Decentralized Local Controller (DLC) is collision-free and complete.
That is, (i) no vertex or edge-swap conflicts occur, and (ii) every agent
reaches its goal in finite time and remains there thereafter.
\end{proposition}

\begin{IEEEproof}
Collision-freeness follows from the per-vertex FIFO authorization rule.
At any time step, at most one agent can be authorized to enter a vertex,
which prevents vertex conflicts. Edge-swap conflicts are excluded as follows. A swap on $\{u,v\}$ at time $t$ would require two simultaneous moves $u\!\to\! v$ and $v\!\to\! u$, hence that the corresponding agents are simultaneously authorized to enter $v$ and to enter $u$, i.e., they are the first elements of $Q_v$ and $Q_u$ at the same time. But the queues are constructed by appending agents in priority order for every vertex visit, hence for any two agents that both visit $u$ and $v$, their relative order is identical in $Q_u$ and in $Q_v$.
Therefore, the lower-priority agent cannot become first in one of these queues while the higher-priority agent is still pending, so opposite traversals cannot be simultaneously authorized.

Completeness follows from Assumption~\ref{ass:residual_reach}.
Each agent has a finite geometric path in the residual graph and all higher-priority agents eventually remain parked at their goals. Consequently, every queue
$Q_v$ is finite and progresses monotonically. Any waiting induced by local
contention is finite, and each agent advances along its path in finite time
until reaching its goal, where it remains permanently.
\end{IEEEproof}

\begin{algorithm}
\caption{\textsc{DLC-Execute}: local queue-based execution}
\label{alg:dlc_execute}
\begin{algorithmic}[1]
\REQUIRE Paths $\mathcal{P}=\{p_1,\dots,p_k\}$ and queues $\{Q_v\}_{v\in V}$
\ENSURE Collision-free trajectories $\mathcal{T}=\{\tau_1,\dots,\tau_k\}$

\STATE For each agent $a_i$: set $r_i\gets 0$ and $\tau_i(0)\gets v_{i,0}$
\STATE $t\gets 0$
\WHILE{$\exists i \text{ with } r_i<\ell_i$}
    \STATE $t\gets t+1$
    \FOR{$i=1$ \TO $k$}
        \IF{$r_i=\ell_i$}
            \STATE $\tau_i(t)\gets g_i$
        \ELSE
            \STATE $u \gets v_{i,r_i}$;\quad $v \gets v_{i,r_i+1}$
            \IF{$a_i$ is the first element of $Q_v$}
                \STATE $\tau_i(t)\gets v$;\quad $r_i\gets r_i+1$
                \STATE Pop the first element of $Q_v$ \COMMENT{consume authorization to enter $v$}
            \ELSE
                \STATE $\tau_i(t)\gets u$
            \ENDIF
        \ENDIF
    \ENDFOR
\ENDWHILE
\RETURN $\mathcal{T}$
\end{algorithmic}
\end{algorithm}

\section{Experimental Evaluation}

This section evaluates the proposed framework. We first compare against representative state-of-the-art MAPF solvers, we then analyze how the behavior of the proposed method varies with its internal design choices. In our simulator, time advances in discrete ticks and each agent executes exactly one action (move or wait) per tick. All simulations have been performed on a workstation equipped with an Intel Core i7 CPU (3GHz), 16GB of RAM, running Ubuntu 22.04 LTS.

\textbf{Experimental Setup and Baselines.} Experiments are conducted on two standard benchmark maps: (i) \emph{Paris\_1\_256}, a large open-space $256 \times 256$ grid that primarily stresses scalability, and (ii) \emph{room-64-64-8}, a $64 \times 64$ map composed of rooms connected by narrow corridors, which induces frequent conflicts. For each map and agent count $k$, start vertices $s_i$ and goal vertices $g_i$ are sampled uniformly at random, subject to being pairwise distinct. The proposed method is compared against the following baselines:
\begin{itemize}
    \item \emph{ECBS} \cite{EECBS}: a bounded-suboptimal centralized solver based on Conflict-Based Search.
    \item \emph{Cooperative A* (CA*)} \cite{P-MAPF}: a classic prioritized planner using a space--time
    reservation table.
    \item \emph{PIBT} \cite{PIBT}: a fast decoupled method based on priority inheritance and backtracking.
\end{itemize}

We report three standard performance metrics:
(i) \emph{Success Rate (SR)}, defined as the percentage of instances solved within a fixed time limit;
(ii) \emph{Runtime}, defined as the wall-clock time (in seconds) required to compute a solution;
and (iii) \emph{Sum-of-Costs (SOC)}, defined as the total number of actions executed by all agents
until they reach their goals.

\begin{figure}[h!]
    \centering
    \includegraphics[width=0.9\columnwidth]{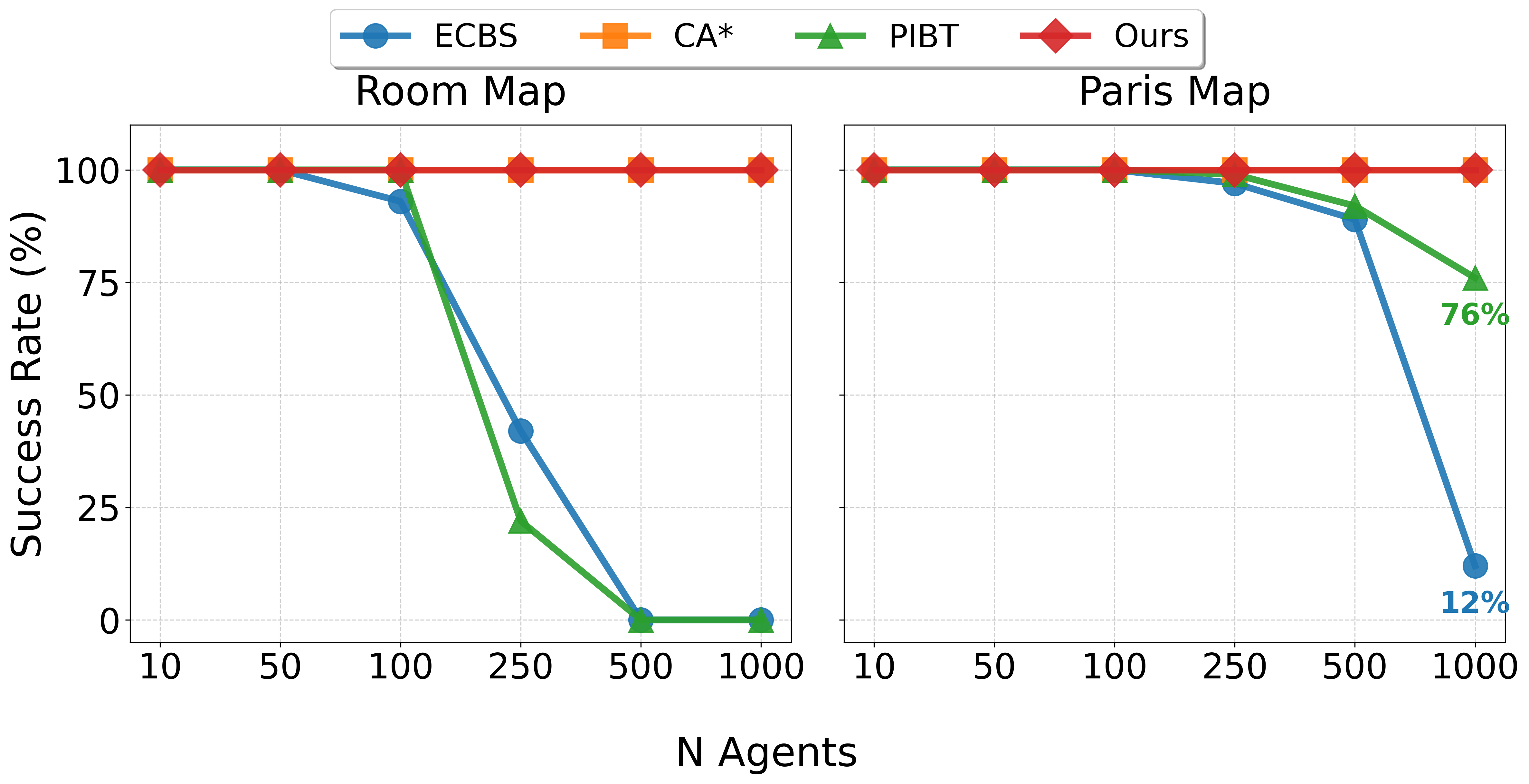}
    \caption{Success rate comparison against methods \cite{EECBS}, \cite{P-MAPF}, \cite{PIBT} (CA* values are consistently 100\% and are covered by the red line of our method) (Mean values over 100 iterations).}
    \label{fig:SR_combined}
\end{figure}

\textbf{Comparison with State-of-the-Art Methods.} The proposed method achieves a $100\%$ success rate on all instances under Assumption~\ref{ass:residual_reach} (see Figure~\ref{fig:SR_combined}), remaining stable even
for large teams in the bottleneck-heavy map. In contrast, ECBS and PIBT exhibit a marked degradation
in success rate on \emph{room-64-64-8} at high agent densities, reflecting the difficulty of resolving
conflicts under severe congestion.

\begin{figure}[h!]
    \centering
    \includegraphics[width=0.9\columnwidth]{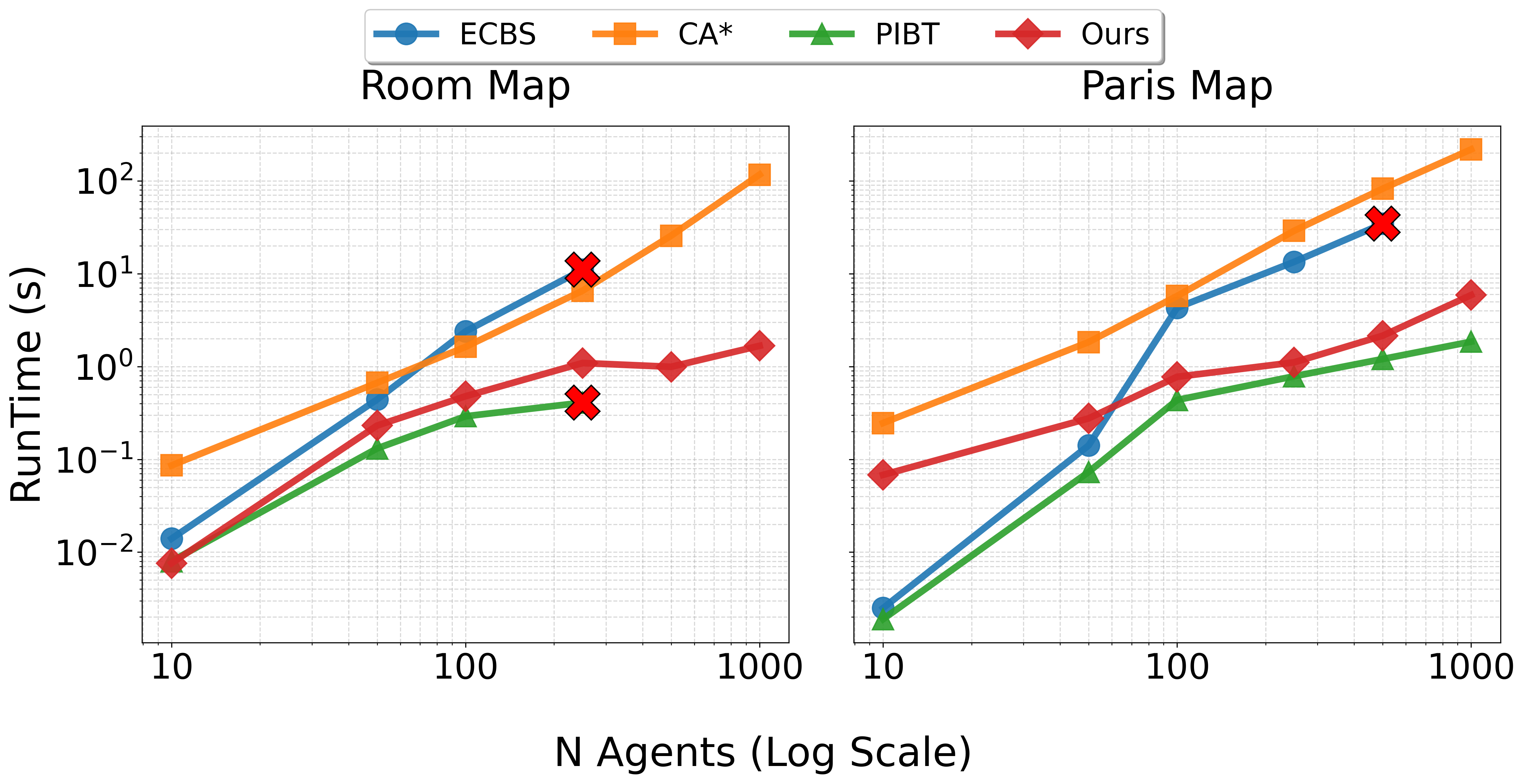}
    \caption{Runtime comparison against methods \cite{EECBS}, \cite{P-MAPF}, \cite{PIBT} (both axes in logarithmic scale). The red 'X' indicates that the algorithm failed to complete due to timeout (Mean values over 100 iterations).}
    \label{fig:runtime_combined}
\end{figure}

In terms of runtime, as shown in Figure~\ref{fig:runtime_combined}, the proposed approach exhibits an empirically near-linear scaling trend and achieves the second-best performance overall, slightly behind PIBT. This efficiency stems from avoiding time-expanded graph construction and resolving conflicts through sparse, local queue operations in the DLC stage.

\begin{figure}[h!]
    \centering
    \includegraphics[width=0.9\columnwidth]{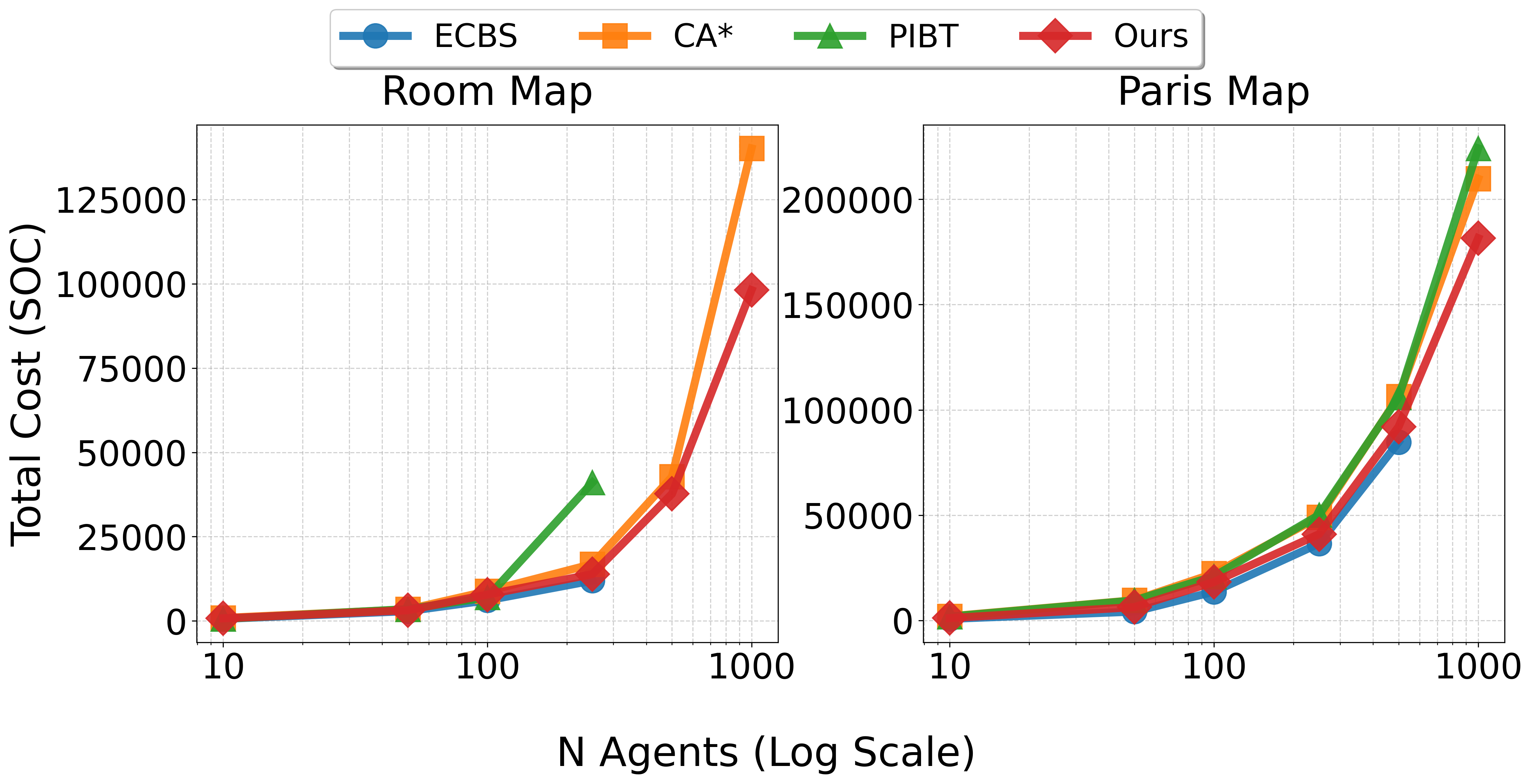}
    \caption{Cost comparison against methods \cite{EECBS}, \cite{P-MAPF}, \cite{PIBT} for Room and Paris maps (Mean values over 100 iterations).}
    \label{fig:SOC_combined}
\end{figure}

Regarding solution quality, as illustrated in Figure~\ref{fig:SOC_combined}, ECBS typically achieves the lowest SOC when it succeeds. Nevertheless,
the proposed method consistently outperforms CA* and PIBT and remains competitive with ECBS.
The advantage is most pronounced in constrained environments, where geometric conflict preemption
significantly reduces execution-time waiting, which is often the dominant contributor to SOC.

\textbf{Impact of Dynamic Costs.} We next analyze the effect of geometric cost inflation through an ablation study. We distinguish between \emph{spatial cost}, defined as the number of vertices in the planned trajectory, and \emph{temporal cost}, representing the total number of wait actions inserted by the DLC stage to resolve conflicts. 
Table~\ref{tab:room_cost_breakdown_ext} compares the full method (\emph{Cost Update}) against an ablated variant (\emph{No Cost Update}) in which the penalty parameter $C_{\mathrm{p}}$ is set to zero. 

Results on the bottleneck-heavy map show that enabling GCP reduces the total cost
in high-density scenarios. At $k=1000$, the total SOC is reduced by $23.1\%$.
This improvement reflects a clear trade-off: a moderate increase in spatial cost due to detours yields a large reduction in temporal cost of about $70\%$.
These results confirm that preemptive geometric detours are significantly cheaper than resolving
repeated conflicts through temporal synchronization.

\begin{table*}[h!] 
    \centering 
    \caption{Spatial and temporal cost analysis comparing the proposed method with and without geometric cost inflation. The results illustrate the trade-off between spatial detours and reduced temporal synchronization (waiting) for collision avoidance (\emph{Room map}, 30.94\% of free space occupied with 1000 agents). Percentages in the first two result columns indicate the relative increase with respect to the alternative configuration.}
    \label{tab:room_cost_breakdown_ext}
    
    \begin{tabular}{l c c c S[table-format=6.0] l l}
        \toprule
        $k$ & \textbf{Method} & {\textbf{Diff. of Spatial Cost}} & {\textbf{Diff. of Temporal Cost}} & {\textbf{Total Cost}} & \textbf{\% of Spatial Cost} & \textbf{\% of Temporal Cost} \\
        \midrule
        
        \multirow{2}{*}{250} & Cost Update & 16\% & {---} & 18518 & 89.8\% & 10.2\% \\
         & No Cost Update & {---} & 67\% & \best{17529} & 82\% & 18\% \\
        \midrule
        
        \multirow{2}{*}{500} & Cost Update & 12\% & {---} & \best{40192} & 76.8\% & 23.2\% \\
         & No Cost Update & {---} & 201\% & 46363 & 59.4\% & 40.6\% \\
        \midrule
        
        \multirow{2}{*}{1000} & Cost Update & 16\% & {---} & \best{106843} & 47.6\% & 52.4\% \\
         & No Cost Update & {---} & 70\% & 138874 & 31.6\% & 68.4\% \\
        \bottomrule
    \end{tabular}
\end{table*}

\textbf{Impact of cost inflating in ordering.} A priority policy defines the ordering $\pi=(a_1,\dots,a_k)$ used by the prioritized planner. 
Each heuristic assigns to every agent $a_i$ a scalar key $\kappa(a_i)$ computed from the static map
and the agent's start/goal pair; agents are then sorted by $\kappa$. Under Assumption~\ref{ass:residual_reach}, the existence of a solution is guaranteed for multiple priority orderings $\pi$. Consequently, the choice of priority policy does not dictate the solvability of the instance, but rather its optimality and the specific spatial-temporal trade-offs made by the GCP, for this reason, a study of the priority heuristics is performed.



Let $d_i$ be the shortest-path distance from $s_i$ to $g_i$. We define the shortest-path corridor $\mathcal{C}_i$ as the set of all nodes lying on any shortest path for agent $a_i$:
\begin{equation}
    \mathcal{C}_i := \{v \in V \mid \mathrm{dist}(s_i,v) + \mathrm{dist}(v,g_i) = d_i\}.
\end{equation}
We define the conflict score $c_i$ based on a single \emph{nominal shortest path} $\hat{p}_i$. The score $c_i$ measures the potential interference this nominal path faces from the corridors of all other agents:
\begin{equation}
    c_i := \sum_{v \in \mathrm{nodes}(\hat{p}_i)} \Big| \{j \neq i \mid v \in \mathcal{C}_j\} \Big|.
\end{equation}
Using $d_i$ and $c_i$, we evaluate five priority policies: \emph{SPF (Shortest Path First)}, which sorts agents by increasing $d_i$; \emph{LPF (Longest Path First)}, by decreasing $d_i$; \emph{CF (Conflicting Path First)}, by decreasing $c_i$; \emph{CL (Conflicting Path Last)}, by increasing $c_i$; and \emph{Random}, using a random permutation.

Table~\ref{tab:priority_comparison} evaluates the influence of different priority heuristics on solution quality. Across both maps, the CL heuristic achieves the lowest SOC in high-density scenarios.
This effect is particularly strong in the constrained map, where prioritizing agents with
immediate conflicts allows cleaner paths to be fixed early, leaving more constrained agents
to be regulated through the resulting vertex queues.
For example, at $k=1000$ on \emph{room-64-64-8}, LPF yields an SOC of $117995$, which is $18\%$ higher
than the $99923$ achieved by CL.

\begin{table*}[h!]
    \centering
    \caption{Comparison of results for different priority policies across maps and agent counts. Percentage values indicate the relative increase of each metric with respect to the best-performing policy.}
    \label{tab:priority_comparison}
    
    \begin{tabular}{l c c c c | l c c c c}
        \toprule
        \multicolumn{5}{c}{\textbf{Paris Map}} & \multicolumn{5}{c}{\textbf{Room Map}} \\
        \midrule
        $k$ & \textbf{Priority} & \textbf{Diff. Solved Colls \%} & \textbf{Diff Cost \%} & \textbf{Total Cost} & 
        $k$ & \textbf{Priority} & \textbf{Diff. Solved Colls \%} & \textbf{Diff Cost \%} & \textbf{Total Cost} \\
        \midrule
        
        \multirow{5}{*}{100} & SPF & 100\% & --- & \best{18092} & 
        \multirow{5}{*}{100} & SPF & 9\% & 4\% & 6982 \\
         & LPF & 335\% & 6\% & 19211 & 
         & LPF & 13\% & 2\% & 6851 \\
         & CF & 40\% & 9\% & 19723 & 
         & CF & 16\% & 12\% & 7541 \\
         & CL & 110\% & 2\% & 18446 & 
         & CL & --- & --- & \best{6723} \\
         & Random & --- & 12\% & 20356 & 
         & Random & 23\% & 11\% & 7448 \\
        \midrule
        
        \multirow{5}{*}{500} & SPF & 9\% & 7\% & 98811 & 
        \multirow{5}{*}{500} & SPF & 12\% & 4\% & 39135 \\
         & LPF & 7\% & 9\% & 100617 & 
         & LPF & 62\% & 16\% & 43799 \\
         & CF & 6\% & 11\% & 101003 & 
         & CF & 41\% & 15\% & 43496 \\
         & CL & --- & --- & \best{92367} & 
         & CL & --- & --- & \best{37822} \\
         & Random & 12\% & 8\% & 99552 & 
         & Random & 28\% & 8\% & 40914 \\
        \midrule
        
        \multirow{5}{*}{1000} & SPF & 9\% & 4\% & 189337 & 
        \multirow{5}{*}{1000} & SPF & 8\% & 6\% & 105546 \\
         & LPF & 13\% & 6\% & 192448 & 
         & LPF & 19\% & 18\% & 117995 \\
         & CF & 11\% & 5\% & 191862 & 
         & CF & 17\% & 15\% & 114721 \\
         & CL & --- & --- & \best{181737} & 
         & CL & --- & --- & \best{99923} \\
         & Random & 16\% & 8\% & 195453 & 
         & Random & 5\% & 3\% & 102744 \\
        \bottomrule
    \end{tabular}
    
\end{table*}

\section{Conclusions} \label{sec:conc}

We presented a prioritized MAPF framework that separates geometric planning from execution-time coordination. Geometric Conflict Preemption discourages overlap with higher-priority paths through a one-shot cost inflation on the original graph without time modeling. The Decentralized Local Controller then enforces collision-free execution through local FIFO queues and introduces waiting only under contention.
The evaluation shows that this decomposition yields stable scalability at large team sizes while preserving competitive SOC, with the largest gains arising in bottleneck-heavy environments where waiting dominates the cost. Future work includes adaptive tuning of the inflation parameter, online priority re-ordering based on observed contention.

\end{document}